\documentclass[11pt]{amsart}
\usepackage{amsmath,amssymb} 
\usepackage{graphicx}
\usepackage[utf8]{inputenc}
\usepackage[font=small,labelfont=bf]{caption}

\begin{document}

\newtheorem{thm}{Theorem}[section]
\newtheorem{lem}[thm]{Lemma}
\newtheorem{prop}[thm]{Proposition}
\newtheorem{coro}[thm]{Corollary}
\newtheorem{defn}[thm]{Definition}
\newtheorem*{remark}{Remark}

\numberwithin{equation}{section}

\newcommand{\Z}{{\mathbb Z}} 
\newcommand{\Q}{{\mathbb Q}}
\newcommand{\R}{{\mathbb R}}
\newcommand{\C}{{\mathbb C}}
\newcommand{\N}{{\mathbb N}}
\newcommand{\FF}{{\mathbb F}}
\newcommand{\T}{{\mathbb T}}
\newcommand{\fq}{\mathbb{F}_q}

\newcommand{\fixmehidden}[1]{}

\def\scrA{{\mathcal A}}
\def\scrB{{\mathcal B}}
\def\cI{{\mathcal I}}
\def\scrD{{\mathcal D}}
\def\cL{{\mathcal L}}
\def\cM{{\mathcal M}}
\def\cN{{\mathcal N}}
\def\scrR{{\mathcal R}}
\def\scrS{{\mathcal S}}

\newcommand{\rmk}[1]{\footnote{{\bf Comment:} #1}}

\renewcommand{\mod}{\;\operatorname{mod}}
\newcommand{\ord}{\operatorname{ord}}
\newcommand{\TT}{\mathbb{T}}
\renewcommand{\i}{{\mathrm{i}}}
\renewcommand{\d}{{\mathrm{d}}}
\renewcommand{\^}{\widehat}
\newcommand{\HH}{\mathbb H}
\newcommand{\Vol}{\operatorname{vol}}
\newcommand{\area}{\operatorname{area}}
\newcommand{\tr}{\operatorname{tr}}
\newcommand{\norm}{\mathcal N} 
\newcommand{\intinf}{\int_{-\infty}^\infty}
\newcommand{\ave}[1]{\left\langle#1\right\rangle} 
\newcommand{\E}{\mathbb E}
\newcommand{\Var}{\operatorname{Var}}
\newcommand{\Cov}{\operatorname{Cov}}
\newcommand{\Prob}{\operatorname{Prob}}
\newcommand{\sym}{\operatorname{Sym}}
\newcommand{\disc}{\operatorname{disc}}
\newcommand{\CA}{{\mathcal C}_A}
\newcommand{\cond}{\operatorname{cond}} 
\newcommand{\lcm}{\operatorname{lcm}}
\newcommand{\Kl}{\operatorname{Kl}} 
\newcommand{\leg}[2]{\left( \frac{#1}{#2} \right)}  
\newcommand{\id}{\operatorname{id}}
\newcommand{\beq}{\begin{equation}}
\newcommand{\eeq}{\end{equation}}
\newcommand{\diam}{\operatorname{diam}}
\newcommand{\sgn}{\operatorname{sgn}}

\newcommand{\sumstar}{\sideset \and^{*} \to \sum}

\newcommand{\LL}{\mathcal L} 
\newcommand{\sumf}{\sum^\flat}
\newcommand{\Hgev}{\mathcal H_{2g+2,q}}
\newcommand{\USp}{\operatorname{USp}}
\newcommand{\conv}{*}
\newcommand{\dist} {\operatorname{dist}}
\newcommand{\CF}{c_0} 
\newcommand{\kerp}{\mathcal K}

\newcommand{\gp}{\operatorname{gp}}
\newcommand{\Area}{\operatorname{Area}}

\title[Non-gaussian waves in \u{S}eba's billiard]{Non-gaussian waves in \u{S}eba's billiard}
\author{P\"ar Kurlberg}
\address{Department of Mathematics, KTH Royal Institute of Technology,
SE-100 44 Stockholm, Sweden}
\email{kurlberg@math.kth.se}
\author{Henrik Uebersch\"ar}
\address{Institut de Math\'ematiques de Jussieu, Campus Pierre et Marie Curie, Sorbonne Universit\'e, 4 place Jussieu, 75005 Paris, France.}
\email{henrik.ueberschar@imj-prg.fr}
\thanks{P.K. was partially supported by
the
Swedish Research Council (2016-03701). H.U. was partially supported by the grant ANR-17-CE40-0011-01 of the French National
Research Agency ANR (project SpInQS)}
\date{\today}
\maketitle

\begin{abstract}
The \u{S}eba billiard, a rectangular torus with a point scatterer, is
a popular model to study the transition between integrability and
chaos in quantum systems. Whereas such billiards are classically
essentially integrable, they may display features such as quantum
ergodicity \cite{KU} which are usually associated with quantum systems
whose classical dynamics is chaotic. \u{S}eba proposed that the
eigenfunctions of toral point scatterers should also satisfy Berry's
random wave conjecture, which implies that the value distribution 
of the eigenfunctions ought to be Gaussian.  

We prove a conjecture of Keating, Marklof and Winn who suggested that
\u{S}eba billiards with irrational aspect ratio violate the random
wave conjecture. More precisely, in the case of diophantine tori, we
construct a subsequence of the set of new eigenfunctions having
even/even symmetry, of essentially full density, and show that its
fourth moment is not consistent with a Gaussian value distribution.
In fact, given any set $\Lambda$ interlacing with the set of
unperturbed eigenvalues, we show non-Gaussian value
distribution of the Green's functions $G_{\lambda}$, for $\lambda$ in
an essentially full density subsequence of $\Lambda$.
\end{abstract}

\section{Introduction}
\fixmehidden{Address Stephane's comment: we rely very much on previous
  papers - recall more?}

\u{S}eba's billiard, a rectangular billiard $\cM$ with irrational
aspect ratio and a Dirac mass placed in its interior, is a popular
model in the field of Quantum Chaos to investigate the transition
between chaos and integrability in quantum systems. The model was
originally proposed by Petr \u{S}eba in 1990 \cite{S} and has since
attracted much attention in the literature\fixmehidden{Add Shigehara etc
  etc? I.e. more physic lit?}
	\cite{CdV,BGS,BGS2,KBB,KMW,K,RU,Y,KU,Y2,KR,KU2}.
%
Although, the Dirac mass only affects a measure zero subset of
trajectories in phase space and thus has essentially no effect on the
classical dynamics, \u{S}eba argued that the wave functions of the
associated quantized billiard may display similar features as quantum
systems which are classically chaotic.

In particular, \u{S}eba conjectured that the wave functions should
obey Berry's random wave model, i. e. be well approximated by
a superposition of monochromatic  random waves as the
eigenvalue tends to infinity. A consequence of this conjecture
is that the moments of the eigenfunctions should converge to the Gaussian
moments in the limit as the eigenvalue tends to infinity (cf. \cite{B}, p. 240, eqs. (78-80)).
In particular, denoting an $L^2$-normalized (real) wave
function with eigenvalue $\lambda$ by $\psi_\lambda$, one expects that
the fourth moment of $\psi_\lambda$
(possibly after excluding a zero density  subsequence of exceptional
eigenvalues) 
converge to the corresponding Gaussian moment as
$\lambda\to\infty$, namely that
$$\E(\psi_\lambda^4)=\int_{\cM} \psi_\lambda^4 d\mu \to 
3,
$$
where $d\mu=d\mu_{Leb.}/\Vol(\cM)$ denotes the normalized Lebesgue measure.

\u{S}eba calculated the value distribution for high energy wave
functions and found seemingly strong numerical evidence for a Gaussian
value distribution in line with Berry's predictions. Later Keating, Marklof and Winn
cast doubt on \u{S}eba's conjecture when they
showed that quantum star graphs, a model known to be  similar in
behaviour to \u{S}eba's billiard, did indeed violate the random wave
model \cite{K,KMW}. 

In this paper we put this matter to rest by showing that for a
\u{S}eba billiard with diophantine aspect ratio (a 
condition that holds generically), the fourth
moment of the eigenfunction cannot tend to a Gaussian. In fact we can find a
subsequence of arbitrarily high density such that the moment stays
bounded away from the Gaussian moment as the eigenvalue tends to
infinity.
In fact, our results are valid for any sequence of numbers which
interlace with the Laplace eigenvalues, in particular for the new
eigenvalues of the both the weak and the strong coupling quantizations
of the \u{S}eba billiard.  The former arises from von Neumann's theory
of self-adjoint extensions, whereas the latter, investigated
numerically in \u{S}eba's paper, uses a different renormalization
which is considered more physically relevant (cf. \cite{U} for a
detailed discussion of weak and strong coupling
quantizations.)

\subsection{Background}

Before we state the results, let us recall the mathematical definition
of \u{S}eba's billiard. In this paper we will mainly focus on periodic
boundary conditions (the case of Dirichlet boundary conditions is
treated in the Appendix) and thus deal with a flat 2-torus
$\T^2=\R^2/2\pi\cL$, where $\cL=\Z(a,0)\oplus\Z(0,1/a)$ for some $a>0$
such that $a^4$ is a diophantine number
(cf. \cite[Ch. 2.8]{katok-hasselblatt}). The formal Schr\"odinger 
operator associated with a Dirac mass placed at the point $x_0\in\T^2$
is given by $$-\Delta+\alpha\delta_{x_0}.$$ This formal operator may
be associated with a one-parameter family of self-adjoint extensions
of the restricted positive Laplacian
$-\Delta|_{C^\infty_c(\T^2-\{x_0\})}$. For the details of this theory
we refer the reader to the introduction and appendix of the paper
\cite{RU}. We adopt the notation of this paper and refer to the
self-adjoint extensions as $-\Delta_\varphi$, where
$\varphi\in(-\pi,\pi)$ is the extension parameter.

One of the key features of the spectral theory of the operator
$-\Delta_\varphi$ is that it represents a rank-one perturbation of the
Laplacian. That is, for each Laplace eigenspace the perturbation ``tears
off'' a new eigenvalue, and the spectrum of $-\Delta_\varphi$ therefore
consists of two parts: the ``old'' and the ``new'' eigenvalues. The
multiplicity of each old eigenvalue is reduced by one and the
corresponding eigenspace is just the co-dimension one subspace of
Laplace eigenfunctions which vanish at $x_0$. This part of the
spectrum is therefore not affected by the presence of the Dirac
mass. On the other hand, the new part of the spectrum ``feels'' the
presence of the scatterer and the 4th moment of these ``new
eigenfunctions'' will be the focus of this paper.

The new eigenvalues interlace with the old Laplace eigenvalues and the
associated eigenfunctions are just Green's functions which have the
following $L^2$-expansion:
$$
G_\lambda(x,x_0)
=\sum_{\xi\in\cL}
\frac{e^{\i\left\langle   \xi,x-x_0 \right\rangle}}{|\xi|^2-\lambda},
$$ 
with the following formula for the 2nd moment:
$$\int_{\T^2}|G_\lambda(x,x_0)|^2 d\mu(x)=\sum_{\xi\in\cL}
\frac{1}{(|\xi|^2-\lambda)^2},$$
where $d\mu(x)=dx/(4\pi^2)$ denotes the normalized Lebesgue measure on
$\T^2$.

The set of new eigenvalues
can be determined as the solutions of a spectral equation \cite{RU}. There is in fact another quantization condition
--- known as a strong coupling quantization --- which is considered
more relevant in the physics literature and requires a renormalization
of the self-adjoint extension parameter $\varphi$ as the eigenvalue
$\lambda$ increases. This leads to a different spectral equation, but
as our results will in fact hold for any sequence which interlaces
with the unperturbed eigenvalues we will not dwell on this matter
(details can be found \cite{U,Sh}.)

\subsection{Results}

Let us denote by $g_\lambda=G_\lambda/\|G_\lambda\|_2$ the
$L^2$-normalized new eigenfunctions. The following theorem is our main
result and shows that the fourth moment of eigenfunctions of
\u{S}eba's billiard is not Gaussian, in particular that the value
distribution of the wave functions is not consistent with a Gaussian
distribution in the limit as the eigenvalue $\lambda$ tends to infinity --- a
contradiction to Berry's random wave model.

\begin{thm}\label{moment}
  Consider a 2-torus with diophantine aspect ratio, and let $\Lambda
  \subset \R$  denote any subset interlacing with the set of
  unperturbed eigenvalues.
  Given
  $\epsilon \in (0,1)$ there exists a subsequence of $\Lambda$, of
  relative density $1-\epsilon$, and a constant $C_\epsilon>0$ such
  that
for $\lambda$ tending to infinity along said subsequence we have
$$1-o(1)\leq\E(g_\lambda^4)\leq3-C_\epsilon+ o(1).$$
\end{thm}

\begin{remark}
  As the torus is homogenous we may place the scatterer at $0$ and it
  is then natural to desymmetrize with respect to odd/even-ness
  vis-a-vis horizontal and vertical reflections.  The set of new
  eigenfunctions is then exactly the eigenspaces having
  even/even-invariance.  Thus, within this symmetry class, essentially
  all  of the eigenfunctions have non-Gaussian fourth moments.
\end{remark}



\subsection{Acknowledgements} We would like to thank Zeev Rudnick for
very helpful discussions and comments on the manuscript.  We would
also like to thank St\'ephane Nonnenmacher for reading the manuscript
very carefully and giving many detailed comments improving the
exposition as well as the result --- in particular for raising the
question whether the results holds for any sequence interlacing
with the unperturbed eigenvalues.

\section{Approximating the 4th moment}\label{sec:approx}

\subsection{$L^4$ convergence}
Let $\cL$ be an irrational rectangular unimodular lattice and consider
the 2-torus $\T^2=\R^2/2\pi\cL$. Fix $\lambda>0$ a new eigenvalue. We
define $c_\lambda(\xi):=(|\xi|^2-\lambda)^{-1}$ and we take
$L=L(\lambda)$ to be an increasing function such that $L\to+\infty$ as
$\lambda\to+\infty$.

The following expansion for the Green's function holds in the
$L^2$-sense:
\begin{equation}
  \label{eq:G-of-x-definition}
G_\lambda(x):=G_{\lambda}(x,0) =
\sum_{\xi\in\cL}c_\lambda(\xi)e^{i\xi\cdot x}
\end{equation}
(without loss of generality we  may assume that $x_{0}=0$.)
Our aim is, first of all, to show that this expansion also holds in
the $L^4$-sense. 
We thus introduce the truncated Green's function
$$G_\lambda^T(x)=\sum_{\xi\in\cL, |\xi|\leq T}c_\lambda(\xi)e^{
  i\xi\cdot x}, \quad T = T(\lambda) := 10\lambda^{1/2},  $$ 
and show that $G_\lambda^T$ converges in
$L^4(\T^2)$, as $T\to\infty$. 

We will achieve this by showing that $G_\lambda^T$ is Cauchy in
$L^4(\T^2)$, in particular we will bound the $L^4$-norm of
the difference $G_\lambda^{2T}-G_\lambda^T$. 
Letting
$$A(T):=\{v\in\cL : |v|\in[T,2T]\}$$ we then have (recall
$T= 10\lambda^{1/2}$, and thus $c_\lambda(v)>0$ for $v\in
A(T)$)\footnote{We denote by $f\ll g$ that there exists a
   constant $C>0$ s.t. $f\leq Cg$.} 
\begin{equation}
\begin{split}
\int_{\T^2}|G_\lambda^{2T}(x)-G_\lambda^T(x)|^4d\mu(x)=&\sum_{v_1,v_2,v_3,v_4\in
  A(T):\sum_{i=1}^{4} v_i=0}\prod_{i=1}^{4} c_\lambda(v_i)\\ 
\ll& \sum_{v_1,v_2,v_3,v_4\in A(T):\sum_{i=1}^{4} v_i=0}\frac{1}{|v_1|^2|v_2|^2|v_3|^2|v_4|^2}\\
\ll& \,\frac{1}{T^8}\cdot|\{v_1,v_2,v_3,v_4\in A(T) : \sum_{i=1}^{4} v_i=0\}|
\end{split}
\end{equation}
and, since $v_4=-\sum_{i=1}^3v_i$, we find that the number of 4-tuples
is at most $|A(T)|^3\ll (T^2)^3$, and thus the above is $\ll
\frac{1}{T^2}.$ Hence $$\|G_\lambda^{2T}-G_\lambda^T\|_4\ll
T^{-1/2}$$ and,
similarly, $$\|G_\lambda^{2^{k+1}T}-G_\lambda^{2^kT}\|_4\ll
2^{-k/2}T^{-1/2}$$ which implies for any integers $p>q>0$ 
$$\|G_\lambda^{2^pT}-G_\lambda^{2^qT}\|_4\ll
T^{-1/2}\sum_{k=q}^{p-1}2^{-k/2}\ll T^{-1/2}2^{-q/2}.$$ 
This then implies, by a telescopic summation, that
$(G_\lambda^{2^{q}T})_q$ is a Cauchy 
sequence and
therefore converges to a limit in $L^4$ as 
$q\to\infty$.
An argument similar to the one used above shows 
that if $\tilde{T} \in [2^{k}T, 2^{k+1}T]$ then
$\|G_\lambda^{2^kT}-G_\lambda^{\tilde{T}}\|_4\ll T^{-1/2}2^{-k/2}$,
and thus $(G_\lambda^{T})_T$ is also a Cauchy sequence.

In particular, we have 
\begin{equation}
  \label{eq:fourth-moment}
\|G_\lambda\|_4^4=
\sum_{v_1,v_2,v_3\in\cL}c_\lambda(v_1)c_\lambda(v_2)c_\lambda(v_3)c_\lambda(v_1+v_2-v_3).
\end{equation}
\subsection{Further truncations}

Let $A(\lambda,L)$ denote the annulus
$$
A(\lambda,L) := \{ v \in \cL : |v|^{2} \in [\lambda-L, \lambda+L]  \}
$$

We introduce the Green's function truncated to lattice points inside
the annulus $A(\lambda,L)$
\begin{equation}
G_{\lambda,L}(x)=\sum_{\xi\in
  A(\lambda,L)}c_\lambda(\xi)e^{i\xi\cdot x}, \quad
c_\lambda(\xi):=\frac{1}{|\xi|^{2}-\lambda}. 
\end{equation}
We have the following lemma which shows that $G_{\lambda,L}$
approximates $G_\lambda$ in $L^4(\T^2)$ as $\lambda\to\infty$
if  $L$ is any growing function of $\lambda$
which tends to infinity.
\begin{lem}\label{L4trunc}
Let $L = L(\lambda) \ge 10$ be an
increasing function 
that tends to infinity 
with $\lambda$.   
There exists a full density subsequence of new eigenvalues such that,
for all $\lambda$ in said subsequence, we have
$$\|G_\lambda-G_{\lambda,L}\|_4\ll L^{-1/12+o(1)}.$$
\end{lem}
\begin{proof}
Let $A_+ = A_+(\lambda,L)$ denote the set $\{
v \in \cL : |v|^{2} > \lambda + L 
\}$ and by $A_- = A_-(\lambda,L)$ the disk $\{v \in \cL : |v|^{2} < \lambda - L\}$.
We begin by noting that
$$\|G_\lambda-G_{\lambda,L}\|_4^4=\sum_{ v_{1}, \ldots v_{4} \in
  A_+\cup A_-: \sum_{i=1}^4 v_{i} = 0} \prod_{i=1}^{4}
\frac{1}{|v_{i}|^{2}-\lambda};
$$
writing
$G_{\lambda}-G_{\lambda,L} = \sum_{v \in{ A_+}} + \sum_{v \in{ A_-}}$
and using the $L^{4}$ triangle inequality we can treat large and small
$v$ separately.  We begin by showing that 
$$
\sum_{ v_{1}, \ldots v_{4} \in A_+(\lambda, L) : \sum_{i=1}^4 v_{i} = 0}
\prod_{i=1}^{4}
\frac{1}{|v_{i}|^{2}-\lambda}
$$
is small (for most $\lambda \in \Lambda$) given that $L$ tends to
infinity as  
$\lambda$ grows.  Up to a bounded
combinatorial factor, we may after reordering terms assume that
$|v_{i+1}|^{2} -\lambda \geq |v_{i}|^{2}-\lambda > 0$ for $i=1,2,3$,
hence
$|v_{4}|^{2}-\lambda \geq \prod_{i=1}^{3}(|v_{i}|^{2}-\lambda)^{1/3}$;
on noting that $v_{4}$ is determined by $v_{1}, v_{2}, v_{3}$, it is
enough to show that
$$
\prod_{i=1}^{3}
\left(
  \sum_{|v_{i}|^{2} \geq \lambda + L  }
  \frac{1}{(|v_{i}|^{2}-\lambda)^{4/3}}
\right)
= o(1).
$$
In particular,  it is enough to show that
$
\sum_{|v|^{2} \geq \lambda + L  }
\frac{1}{(|v|^{2}-\lambda)^{4/3}} = o(1);
$
which in turn reduces to showing  that
$$
\sum_{2 \lambda \geq  |v|^{2} \geq \lambda + L  }
\frac{1}{(|v|^{2}-\lambda)^{4/3}} = o(1)
$$
(to see this, use Weyl's law and  partial summation to bound the
contribution from $v$ such that $|v|^{2} > 2 \lambda$.)

Now, given an integer $k \geq 0$, let $M(k)$ denote the number of
unperturbed eigenvalues in the interval
$[k,k+1)$, or equivalently, the number of lattice points $v$ such that
$|v|^{2} \in [k,k+1)$.
We consider the sum over all $\lambda \in (T/2, T) \cap \Lambda$ , and
show that dyadic means of
$ \sum_{2 \lambda \geq |v|^{2} \geq \lambda + L }
\frac{1}{(|v|^{2}-\lambda)^{4/3}}$ are small for $T$ large. 
More precisely,
$$
\sum_{\lambda \in \Lambda \cap (T/2, T)}
\sum_{2 \lambda \geq |v|^{2} \geq \lambda + L  }
\frac{1}{(|v|^{2}-\lambda)^{4/3}}
\ll
\sum_{l < T}
M(l)
\sum_{\lambda \geq k \geq L} \frac{M(l+k)}{k^{4/3}}
$$
which, using the same argument as in the proof of \cite[Lemma~3]{KU2}
is
$$
\ll
\sum_{\lambda \geq  k \geq L} \frac{1}{k^{4/3}}
\sum_{l<T}
M(l) M(k+l)
\ll
\sum_{\lambda \geq k \geq L} \frac{1}{k^{4/3}}
\sum_{l<T}
M(l)^{2}
\ll
L^{-1/3} T.
$$
Hence, using Chebychev's inequality, for most
 $\lambda \in \Lambda \cap (T/2, T)$ 
we find that 
\begin{equation}
  \label{eq:bound-large-v-sum}
\sum_{v \in \cL : |v|^{2} \geq \lambda + L  }
\frac{1}{(|v|^{2}-\lambda)^{4/3}}
\ll
L^{-1/3+o(1)}.
\end{equation}
A similar argument shows that, for most $\lambda \in \Lambda \cap (T/2,T)$,
\begin{equation}
  \label{eq:bound-small-v-sum}  
\sum_{v \in \cL : |v|^{2} \leq \lambda - L  }
\frac{1}{(|v|^{2}-\lambda)^{4/3}}
\ll
L^{-1/3+o(1)}
\end{equation}
and hence the $L^{4}$ norm is $\ll L^{-1/12+o(1)}$.

\end{proof}

\section{Proof of Theorem \ref{moment}}

One finds (cf. the $L^2$-expansion of the Green's function
(\ref{eq:G-of-x-definition}), or see \cite[eq. (3.22)]{RU}) that
\begin{equation}
\left\|G_{\lambda}\right\|_{2}^{2}=\int_{\T^2}|G_\lambda|^2 d\mu=\sum_{\xi\in\cL}\frac{1}{(|\xi|^{2}-\lambda)^{2}}
=\sum_{n\in\cN}\frac{r(n)}{(n-\lambda)^{2}},
\end{equation}
where $r(n)$ is the multiplicity of the Laplace eigenvalue $n$
and $$\cN=\{n_0=0<n_1<n_2<\cdots\}$$ denotes the set of distinct
(unperturbed) Laplace eigenvalues.

Also
(cf. (\ref{eq:fourth-moment})),
\begin{equation}
\begin{split}
&\left\|G_{\lambda}\right\|_{4}^{4}=\int_{\T^2}|G_\lambda|^4 d\mu=\\
&\sum_{\substack{\xi_{1}+\xi_{2}=\eta_{1}+\eta_{2}\\ \xi_{1},\xi_{2},\eta_{1},\eta_{2}\in\cL}}\frac{1}{(|\xi_{1}|^{2}-\lambda)(|\xi_{2}|^{2}-\lambda)(|\eta_{1}|^{2}-\lambda)(|\eta_{2}|^{2}-\lambda)}.
\end{split}
\end{equation}

\subsection{The sequence $\Lambda_g$}
We recall some useful results from sections 6 and 7 of
\cite{RU}.\fixmehidden{Add more details here? Stephane's comment} Let
$\theta< 1/3$ denote the best known exponent in the error term for the
circle problem for a rectangular lattice \cite{Hu}.  In fact, we will
only need $\theta< 1/2$, just a bit beyond the trivial geometric
estimate.
Adopting the notation of \cite{RU} we let
$\delta\in(0,\tfrac{2}{3}(\tfrac{1}{2}-\theta))$ and
define
$$
S(\lambda)=
\bigcup_{\substack{0\neq\zeta\in\cL\\  |\zeta|<\lambda^{\delta/2}}}S_\zeta.
$$
where we define $S_{\zeta}$ for any $\zeta\in\cL\setminus\{0\}$ as the set of
solutions to a certain diophantine inequality (cf. eq. (6.1) in
\cite{RU}), namely
$$
S_\zeta := \{\eta\in\cL \mid |\left\langle\eta,\zeta\right\rangle|\leq |\eta|^{2\delta}\}.
$$

We will show
that the subset of ``good'' eigenvalues 
$$\Lambda_g:=\{\lambda\in\Lambda \mid A(\lambda,\lambda^\delta)\cap
S(\lambda)=\emptyset\}$$ 
is of full density in $\Lambda$ (recall that $|\{
\lambda \in \Lambda : \lambda \le X \}| \sim X$), 
by   showing that $$\{\lambda\in\Lambda\setminus\Lambda_g \mid
\lambda\leq X\}\ll X^{1-\delta_0}$$ for
$\delta_0=\tfrac{1}{2}-\theta-\tfrac{3}{2}\delta>0$.
To see this,
observe that the complement of $\Lambda_g$, i.e. the set of ``bad''
elements which we denote by $\Lambda_b$, is of the form 
$$\Lambda_b=\Lambda\setminus\Lambda_g=\bigcup_{\substack{0\neq\zeta\in\cL\\
    |\zeta|<\lambda^{\delta/2}}}B_{\zeta}$$ 
where
$B_{\zeta}=\{\lambda\in\Lambda \mid  
A(\lambda,\lambda^\delta)\cap S_\zeta\neq\emptyset\}$. Here we used 
$$\lambda\notin\Lambda_g \Leftrightarrow A(\lambda,\lambda^\delta)\cap
S(\lambda)\neq\emptyset \Leftrightarrow \bigcup_{\substack{0\neq\zeta\in\cL\\
    |\zeta|<\lambda^{\delta/2}}}(A(\lambda,\lambda^\delta)\cap
S_\zeta)\neq\emptyset$$ 
which is  equivalent to
$$ \lambda\in \bigcup_{\substack{0\neq\zeta\in\cL\\
    |\zeta|<\lambda^{\delta/2}}}B_\zeta.$$  

We recall the bound (6.4) in \cite{RU}, namely that, for fixed $\zeta\in\cL$,
$$|\{\lambda\in B_{\zeta}\mid\lambda\leq
X\}|\leq\frac{X^{1/2+\theta+\delta}}{|\zeta|}.$$ (Note that in the
proof of (6.4), the only property used regarding the location of the
$\lambda$'s is the interlacing property.  Further, the lower bound
$ \theta/2 > \delta$, stated at the beginning of \cite[Section~6]{RU},
is not used in order to prove (6.4).)
We may now apply this bound to get an estimate on the number of bad eigenvalues $\lambda\leq X$.
Note that we are summing over lattice vectors $\zeta\in\cL$ which are
not too large, i.e. $|\zeta|<\lambda^{\delta/2}\leq X^{\delta/2}$, and
we find that
$$|\{\lambda\in\Lambda_{b}\mid\lambda\leq X\}|\leq
X^{1/2+\theta+\delta}\sum_{\substack{0\neq\zeta\in\cL\\
    |\zeta|<X^{\delta/2}}}\frac{1}{|\zeta|}\ll
X^{1/2+\theta+3\delta/2}=X^{1-\delta_0}$$
where
 $\delta_0=\tfrac{1}{2}-\theta-\tfrac{3}{2}\delta>0$ (we
stress
that only the condition $0<\delta<\tfrac{2}{3}(\tfrac{1}{2}-\theta)$
is required).

\subsection{Diagonal solutions}
We begin with the following Lemma which shows that if
$\lambda\in\Lambda_g$, then $A(\lambda,L)$ contains only lattice
points that are reasonably well-spaced.  Recall that $a^{2}$ is the 
aspect ratio of the lattice $\cL$.
\begin{lem}\label{spacings}
Let $1\leq\lambda\in\Lambda_{g}$ and
put $L = L(\lambda) := \tfrac{1}{20}\min(a,1/a)\lambda^{\delta/2}$.
If $\xi$ and $\eta$ are two distinct lattice points belonging to $A(\lambda,L)$, then $|\xi-\eta|\geq\lambda^{\delta/2}$.
\end{lem}
\begin{proof}
To see this, put $\beta=\eta-\xi$ and suppose for  contradiction that
$|\beta|=|\eta-\xi|<\lambda^{\delta/2}$.  As $\lambda\in\Lambda_{g}$,
and $\xi\in A(\lambda,L)$ we find that
$$
| |\xi|^{2}-\lambda| = 
||\eta-\beta|^2-\lambda|<L
$$ and after multiplying out we
obtain
$$||\eta|^2-\lambda+|\beta|^2-2\left\langle
  \eta,\beta\right\rangle|<L.$$ Now, since
$|\beta|<\lambda^{\delta/2}$ and $\eta\in A(\lambda,L)$, it follows
$$2|\left\langle \eta,\beta\right\rangle|< ||\eta|^2-\lambda+|\beta|^2|+L\leq ||\eta|^2-\lambda|+|\beta|^2+L< 2L+\lambda^\delta$$
and, since our assumption implies $L<\frac{1}{4}\lambda^{\delta/2}$,
$$
|\left\langle
  \eta,\beta\right\rangle|<L+\tfrac{1}{2}\lambda^{\delta}<\tfrac{1}{4}\lambda^{\delta/2}+\tfrac{1}{2}\lambda^{\delta} 
<\tfrac{3}{4}\lambda^{\delta}\leq
(\tfrac{3}{4})^{1-\delta}|\eta|^{2\delta}\leq |\eta|^{2\delta},
$$ where
we used $\lambda\leq|\eta|^2+L<|\eta|^2+\tfrac{1}{4}\lambda$ and
therefore $\lambda\leq \tfrac{4}{3}|\eta|^2$. 

This shows that $A(\lambda,\lambda^\delta)\cap S_\beta\neq\emptyset$,
for some $\beta\neq 0$ such that $|\beta|<\lambda^{\delta/2}$, which in
turn implies $A(\lambda,\lambda^\delta)\cap S\neq\emptyset$, 
contradicting that $\lambda\in\Lambda_{g}$. So it follows
that $|\beta|\geq\lambda^{\delta/2}$.
\end{proof}

The following key Lemma will be used in the computation of
the fourth moment. 
\begin{lem}\label{case2} Let
  $\lambda\in\Lambda_{g}$, $\lambda^{\delta/2}>2$ and
put $L = L(\lambda) := \tfrac{1}{20}\min(a,1/a)\lambda^{\delta/2}$.
For
  $\xi,\eta\in A(\lambda,L)$ distinct, 
 the equation
\begin{equation}
  \xi-\eta=\eta'-\xi',
  \quad \xi',\eta'\in A(\lambda,L)
\end{equation}
has only the trivial solutions
\begin{equation}
(\xi',\eta')=
\begin{cases}
(\eta,\xi)\\
\\
(-\xi,-\eta).
\end{cases}
\end{equation}
\end{lem}

\begin{proof}
We define the annulus centered at $\omega\in\R^{2}$
by $$\scrA(\omega)= \scrA(\omega,L)=
\left\{x\in\R^{2}\mid||x-\omega|^{2}-\lambda|<L\right\}$$ 
and denote $\scrA=\scrA(0)$, $\scrB=\scrA\cap\cL$. Let
$\eta,\xi\in\scrB$ and denote $\beta=\eta-\xi$.  

We consider the set
\begin{equation}
\scrS(\beta)=\left\{(\eta',\xi')\in\scrB\times\scrB\mid\eta'-\xi'=\beta\right\}
\end{equation}
and prove that
$$\scrS(\beta)=\left\{(\eta,\xi),(-\xi,-\eta)\right\}.$$

First of all we have from Lemma \ref{spacings} that
$|\xi-\eta|,|\xi+\eta|\geq\lambda^{\delta/2}$.
Also note that any element $(\eta',\xi')$ of $\scrS(\beta)$ satisfies
$$\lambda-L<|\eta'|^{2}<\lambda+L$$
and 
$$
\lambda-L< |\xi|^{2} = |\eta'-\beta|^{2}<\lambda+L
$$ and thus $\eta'$ is
constrained to lie in $\scrA\cap\scrA(\beta)\cap\cL$. Rotate
$\scrA$ around the origin such that $\beta$ is
horizontal.

We next show that the intersection of the two annuli cannot have a
single connected component.  To see this let $R=\sqrt{\lambda+L}$,
$r=\sqrt{\lambda-L}$ and note that the case of a single connected
component implies the
inequality $$\sqrt{\lambda-L}=r\leq \frac{1}{2}|\beta|.$$ Suppose, for
a contradiction, that this inequality holds.  Then
$$
\frac{1}{4}|\beta|^2+\frac{1}{4}|\xi+\eta|^2 
=
\frac{1}{4}|\xi-\eta|^2 + \frac{1}{4}|\xi+\eta|^2
=
\frac{1}{2}( |\xi|^{2}+|\eta|^2)
\leq
R^2=\lambda+L.
$$ 
These two inequalities imply, on recalling our assumption
$L = \tfrac{1}{20}\min(a,1/a)\lambda^{\delta/2}$
$$\frac{1}{4}|\eta+\xi|^2\leq
\lambda+L-\frac{1}{4}|\beta|^2\leq 2L< \frac{1}{2}\lambda^{\delta/2}$$ 
and thus $|\eta+\xi|< \sqrt{2}\lambda^{\delta/4}$. But, as we saw above, our assumption $\lambda\in\Lambda_g$ implies $|\eta+\xi|\geq\lambda^{\delta/2}$, which
contradicts the assumption $\lambda^{\delta/2} > 2$.

{\bf The case of two connected components.}
By the above argument, the set
$$\scrA\cap\scrA(\beta)=:\scrD(\eta)\cup\scrD(-\xi)$$
is thus the union of two approximate parallelograms
containing $\eta$ and $-\xi$ respectively (cf. Figure 1.) 

{\bf Finding the solutions.}
We
introduce coordinates $x$, $y$ such that the annulus $\scrA$ is
centered at $(x,y)=(0,0)$ and $\scrA(\beta)$ is centered at
$(x,y)=(|\beta|,0)$. We compute the coordinates of the vertices $\omega_1$, $\omega_2$, $\nu_1$, $\nu_2$ of
$\scrD(\eta)$ in order to calculate the distances
$h=|\omega_1-\omega_2|$ and $w=|\nu_1-\nu_2|$ (cf. Figure
1).

\begin{center}
\includegraphics[scale=.3]{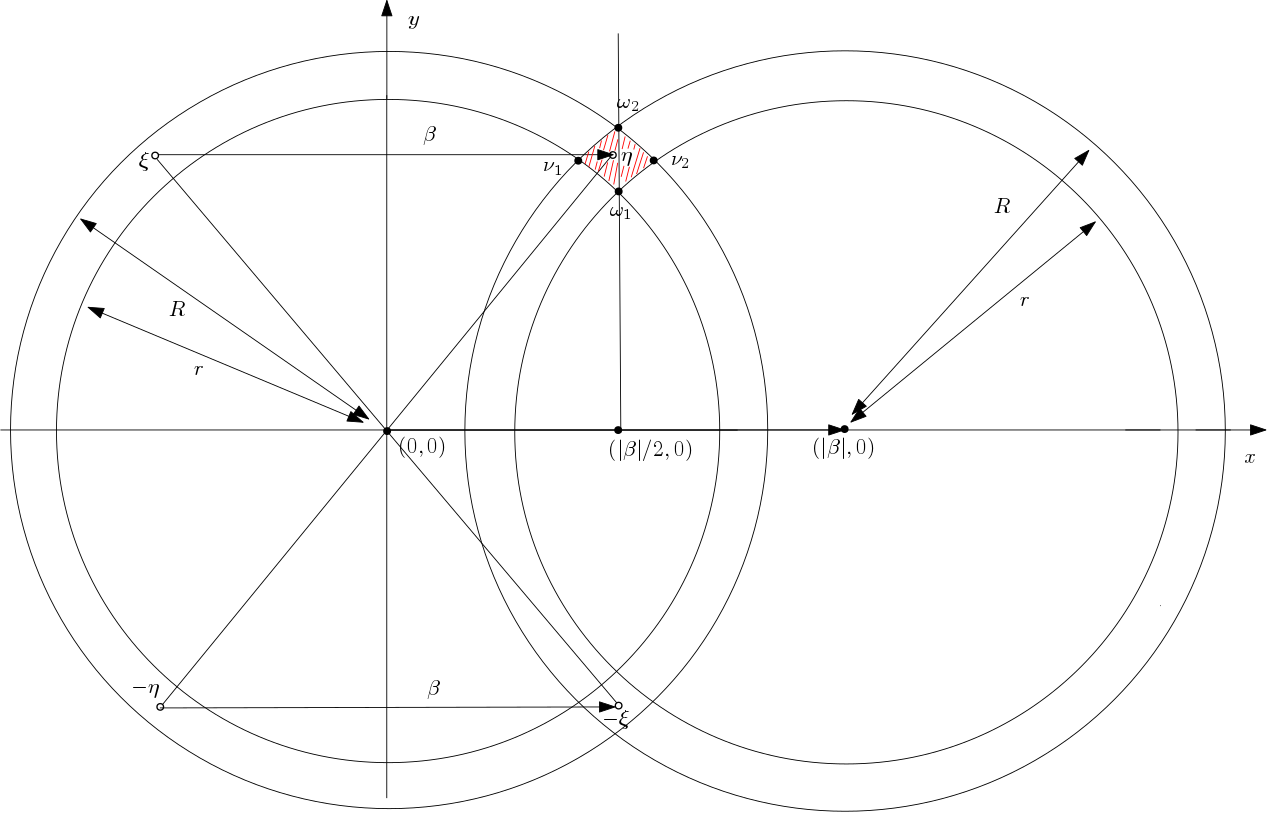}
\captionof{figure}{The intersection of the two annuli $\scrA(0)$ and $\scrA(\beta)$. In order to calculate the diameter of the approximate parallelogram $\scrD(\eta)$ with the vertices $\omega_1$, $\omega_2$, $\nu_1$, $\nu_2$ we have applied a rotation and introduced cartesian coordinates $x$, $y$ such that $\beta=(0,|\beta|)$ in these new coordinates.}
\end{center}

We aim for a bound on the diameter of $\scrD$ which is smaller than
the minimal distance 
between two lattice points, so that $\scrD$ may contain at most one
lattice point. To this end, we observe that $\scrD\subset\scrR=[x_-,x_+]\times[y_r,y_R]$, where $x_-,x_+$ 
are the $x$-coordinates of the points $\nu_1,\nu_2$ and $y_r,y_R$ are the $y$-coordinates of the points $\omega_1,\omega_2$.
We then bound the diameter of $\scrR$.

By solving the equations
$$x^2+y^2=r_1^2 \qquad (x-|\beta|)^2+y^2=r_2^2$$ for the cases
$r_1=r,R$ and $r_2=r,R$, we obtain 
$$\omega_1=(\tfrac{1}{2}|\beta|,y_r), \quad \omega_2=(\tfrac{1}{2}|\beta|,y_R)$$
where $y_r=\sqrt{r^2-\tfrac{1}{4}|\beta|^2}$ and $y_R=\sqrt{R^2-\tfrac{1}{4}|\beta|^2}$. 
It follows that $$h=|\omega_1-\omega_2|=y_R-y_r=\sqrt{R^2-\tfrac{1}{4}|\beta|^2}-\sqrt{r^2-\tfrac{1}{4}|\beta|^2}$$
and therefore (recall $R=\sqrt{\lambda+L}$ and $r=\sqrt{\lambda-L}$)
$$h=\frac{R^2-r^2}{\sqrt{R^2-\tfrac{1}{4}|\beta|^2}+\sqrt{r^2-\tfrac{1}{4}|\beta|^2}}=\frac{2L}{\sqrt{\lambda+L-\tfrac{1}{4}|\beta|^{2}}+\sqrt{\lambda-L-\tfrac{1}{4}|\beta|^{2}}}.$$

Furthermore, by symmetry we have $$\nu_1=(x_-,y_\nu),\quad \nu_2=(x_+,y_\nu)$$ for some $y_\nu>0$ and $x_\pm=\tfrac{1}{2}|\beta|\pm\Delta_\nu$ for some $\Delta_\nu>0$.
We then have $$x_{+}-x_{-}=|\nu_1-\nu_2|=2\Delta_\nu.$$

In order to determine $\Delta_\nu$ we solve the system of
equations $$x_-^2+y_\nu^2=r^2, \quad x_+^2+y_\nu^2=R^2$$ 
which implies $$x_+^2-x_-^2=R^2-r^2.$$ It follows that $2|\beta|\Delta_\nu=R^2-r^2=2L$.
In summary, using that $|\beta| = |\eta - \xi| \ge
\lambda^{\delta/2}$, we find that
$$h=\frac{2L}{\sqrt{\lambda+L-\tfrac{1}{4}|\beta|^{2}}+\sqrt{\lambda-L-\tfrac{1}{4}|\beta|^{2}}}\; 
\text{and}\;w=\frac{2L}{|\beta|}<2\frac{L}{\lambda^{\delta/2}},$$
respectively.
Now, since $0 < L<\tfrac{1}{4\sqrt{2}}\min(a,1/a)\lambda^{\delta/2}$, it follows that $w<\min(a,1/a)/\sqrt{2}$
and $$h\leq\frac{2L}{\sqrt{R^2-\tfrac{1}{4}|\beta|^{2}}}\leq\frac{4L}{\lambda^{\delta/2}}<\frac{\min(a,1/a)}{\sqrt{2}}$$
since
$\tfrac{1}{4}|\beta|^2+\tfrac{1}{4}|\xi+\eta|^2=\tfrac{1}{2}(|\xi|^2+|\eta|^2)\leq
R^2$ and $|\xi+\eta|\geq\lambda^{\delta/2}$. 

Hence
$\diam\scrD(\eta)\leq\diam\scrR(\eta)=\sup_{x,y\in\scrR(\eta)}|x-y|\leq\sqrt{2}\max\{w,h\}<\min(a,1/a)$
and, therefore, 
$\eta$ is the only lattice point in $\scrD(\eta)$. 

By symmetry it follows that $\scrD(-\xi)$ also contains only the lattice point $-\xi$. This proves the claim.
\end{proof}

\subsection{Evaluating the fourth moment}
Recall the truncated Green's function
\begin{equation}
G_{\lambda,L}(x)=\sum_{\xi\in A(\lambda,L)}c_\lambda(\xi)e^{i\xi\cdot x}, \quad c_\lambda(\xi)=\frac{1}{|\xi|^{2}-\lambda}.
\end{equation}
We evaluate the $L^4$-norm of the truncated Green's function in terms
of its $L^2$-norm. 
\begin{lem}\label{mom4}
Let $\lambda\in\Lambda_{g}$ and
and put 
$L  = L(\lambda) := \tfrac{1}{20}\min(a,1/a)\lambda^{\delta/2}$.
Then\footnote{Recall
  that $\E(f)=\int_{\T^2}f(x)d\mu(x)$.} 
$$
\E\left(\frac{G_{\lambda,L}^4}{\left\|G_{\lambda,L}\right\|_2^4}\right)=3-2\frac{\sum_{\xi\in
    A(\lambda,L)}{c_\lambda(\xi)^4}}{\left\|G_{\lambda,L}\right\|_2^4}.
$$ 
\end{lem}
\begin{proof}
Let 
$$a_{\xi}=
\begin{cases}
\frac{1}{|\xi|^{2}-\lambda}, \;\text{if}\; ||\xi|^{2}-\lambda|<L\\
\\
0, \;\text{otherwise;}
\end{cases}
$$ 
clearly $a_{\xi}=a_{-\xi}$. Now
\begin{equation}
\begin{split}
\left\|G_{\lambda,L}\right\|_{4}^{4}
=&\sum_{\substack{\xi_{1},\xi_{2},\eta_{1},\eta_{2}\in\Z^2\\
\xi_{1}+\xi_{2}=\eta_{1}+\eta_{2}}}
a_{\xi_{1}}a_{\xi_{2}}a_{\eta_{1}}a_{\eta_{2}}\\
=&\sum_{0=\xi_{1}-\eta_{1}=\eta_{2}-\xi_{2}}a_{\xi_{1}}a_{\xi_{2}}a_{\eta_{1}}a_{\eta_{2}}\\
&+\sum_{\substack{\beta\neq0\\ \beta=\xi_{1}-\eta_{1}=\eta_{2}-\xi_{2}}}a_{\xi_{1}}a_{\xi_{2}}a_{\eta_{1}}a_{\eta_{2}}.
\end{split}
\end{equation}
The first sum can be rewritten as
\begin{equation}
\sum_{\xi_{1},\xi_{2}}a_{\xi_{1}}^2 a_{\xi_{2}}^2=\left\|G_{\lambda,L}\right\|_{2}^{4}.
\end{equation}
With regard to the second sum let us consider the solutions of the
equation $$\eta_{2}-\xi_{2} = \beta$$
where $$0\neq\beta=\xi_{1}-\eta_{1}$$
and $$\xi_{1},\xi_{2},\eta_{1},\eta_{2}\in A(\lambda,L).$$ Our
assumption that $\lambda\in\Lambda_{g}$, together with
Lemma~\ref{lowerbound}, implies that the only 
solutions are of the form 
\begin{equation}
(\xi_{2},\eta_{2})=
\begin{cases}
(\eta_{1},\xi_{1})\\
\\
(-\xi_{1},-\eta_{1}).
\end{cases}
\end{equation}
Hence, we can rewrite the second sum as
\begin{equation}
2\sum_{\xi_{1},\eta_{1}, \xi_1\neq\eta_1}a_{\eta_{1}}^2 a_{\xi_1}^2=2\left\|G_{\lambda,L}\right\|_{2}^{4}-2\sum_{\xi}a_\xi^4.
\end{equation}
The result follows.
\end{proof}

We have the following Lemma which shows that the 4th moment cannot be
Gaussian, unless the Laplace spectrum has unbounded
multiplicities.\fixmehidden{Maybe we should state that we need
  $L>G_{\epsilon}$?} 
\begin{lem}\label{lowerbound}
  Given $\epsilon \in (0,1)$ there exists a subset of $\Lambda$, of
  density $1-\epsilon$, and a constant $C_\epsilon>0$ such that for
  all sufficiently large
  $\lambda$ in said subsequence, we have 
  $$\frac{\sum_{\xi\in
      A(\lambda,L)}{c_\lambda(\xi)^4}}{\left\|G_{\lambda,L}\right\|_2^4}\geq
  C_\epsilon.$$
\end{lem}
\begin{proof}
We claim that there exists a subsequence of $\Lambda$ of the form $\{\lambda_m\}_{m\in\cN'}$, where $\cN'$ is of density $1-\epsilon$ in $\cN$, such
that a positive proportion of the $L^2$-norm is captured by a finite
set of frequencies in the sense that for 
$I_m:= \cN \cap [m_-
-3,m+3]$, we have
$$\sum_{n}\frac{r_\cL(n)}{(n-\lambda_m)^2}\ll_\epsilon
\sum_{n\in I_m}\frac{r_\cL(n)}{(n-\lambda_m)^2}$$ and,
as $m\to\infty$ along this
subsequence, that $|I_m|$ remains bounded.

Let us explain the construction in more detail. 
In view of the remarks after Lemma 4.2 in \cite{KU2} we may construct
a subsequence
$\cN''$ of density $1-\epsilon$   
such that for $m\in\cN''$ we have 
$$\sum_{|n-m|>3}\frac{r(n)}{(n-m)^2} \leq F_\epsilon,$$
$$\#\{0<|n-m|\leq 3\}\leq E_\epsilon$$ 
and $$|m-\lambda_m|\leq G_\epsilon,$$
for some numbers $E_{\epsilon}, F_{\epsilon},
G_{\epsilon}>0$.

We then have
\begin{equation}
\begin{split}
\sum_n \frac{r(n)}{(n-\lambda_m)^2}
=&\sum_{n\in I_m}\frac{r(n)}{(n-\lambda_m)^2}
+\sum_{n\notin I_m} \frac{r(n)}{(n-\lambda_m)^2}\\
<& \sum_{n\in I_m}\frac{r(n)}{(n-\lambda_m)^2}
+\sum_{|n-m_-|>3} \frac{r(n)}{(n-m_-)^2}
+\sum_{|n-m|>3} \frac{r(n)}{(n-m)^2}
\end{split}
\end{equation}
where we used the inequalities 
$$\sum_{\substack{|n-m_-|>3 \\ n<m_-}}\frac{r(n)}{(n-\lambda_m)^2}<\sum_{\substack{|n-m_-|>3 \\ n<m_-}}\frac{r(n)}{(n-m_-)^2}$$
and
$$\sum_{\substack{|n-m|>3 \\ n>m}}\frac{r(n)}{(n-\lambda_m)^2}<\sum_{\substack{|n-m|>3 \\ n>m}}\frac{r(n)}{(n-m)^2}.$$

So we may define a subsequence $\cN' \subset \cN''$ (of density at least
$1-2\epsilon$) consisting of those $m \in\cN''$ such that {also}
$m_- \in \cN''$ holds.
For  $m \in \cN'$ we have
$$\sum_n \frac{r(n)}{(n-\lambda_m)^2}
< \sum_{n\in I_m}\frac{r(n)}{(n-\lambda_m)^2}+2F_\epsilon$$

The term corresponding to the choice $n=m$ in the short sum (i.e.,
the sum over elements $n \in I_{m}$) is bounded
below by $1/G_\epsilon^2$, so on multiplying by $F_\epsilon
G_\epsilon^2$ we get 
$$
F_\epsilon \leq F_\epsilon G_\epsilon^2 \sum_{n\in I_m} \frac{r(n)}{(n-\lambda_m)^2}
$$
and the number of terms in the short sum is bounded by $2E_\epsilon$.

So we have
$$
\sum_n \frac{r(n)}{(n-\lambda_m)^2} \leq 
(1+2F_\epsilon
G_\epsilon^2)\sum_{n\in I_m} \frac{r(n)}{(n-\lambda_m)^2}
$$ 
and we note that $\#\{n\in I_m\}\leq 2 E_\epsilon$.
This implies
\begin{multline} 
\left(\sum_{|n-\lambda_m|\leq L}\frac{r(n)}{(n-\lambda_m)^2}\right)^2
\leq
\left( \sum_n \frac{r(n)}{(n-\lambda_m)^2} \leq \right)^2
\ll_\epsilon 
\left(\sum_{n\in I_m}\frac{r(n)}{(n-\lambda_m)^2} \right)^2 \\
\ll_{\epsilon} 
|\{n\in I_m\}|\sum_{n\in I_m}\frac{r(n)^{2}}{(n-\lambda_m)^4}
\ll_{\epsilon} 
\sum_{|n-\lambda_m|\leq L}\frac{r(n)}{(n-\lambda_m)^4}
= \sum_{\xi\in A(\lambda,L)} c_{\lambda_m}(\xi)^4
\end{multline}
where we used Cauchy-Schwarz and the fact that the mulitiplicities
$r_\cL(n)$ are bounded (as the aspect ratio of $\cL$ is irrational.)
\end{proof}

It is a simple consequence of the Lemma above that if the
multiplicities in the unperturbed Laplace spectrum are bounded, as is
the case for \u{S}eba's billiard in the irrational aspect ratio
case, then one can construct an {essentially full density}
subsequence of new 
eigenvalues such that the 4th moment does not converge to the Gaussian
4th moment, as the eigenvalue tends to infinity.

\begin{coro}\label{bound_trunc}
  Denote by $g_{\lambda,L}$ the $L^2$-normalized, truncated Green's function on an irrational torus with Diophantine aspect
  ratio, and put
$L  = L(\lambda) := \tfrac{1}{20}\min(a,1/a)\lambda^{\delta/2}$.
  For any $\epsilon >0$, there exists $C_{\epsilon}>0$ and a
  subsequence of $\Lambda$, of density $1-\epsilon$, such that
$$1\leq\E(g_{\lambda,L}^4)\leq 3-2C_\epsilon$$ as $\lambda\to\infty$
along said subsequence.
\end{coro}
\begin{proof}
We recall that there exists a full density subsequence $\Lambda_g$
such that for $\lambda\in\Lambda_g$ we
have $$\frac{\|G_{\lambda,L}\|_4^4}{\|G_{\lambda,L}\|_2^4}=3-2\frac{\sum_{\xi\in 
    A(\lambda,L)}{c_\lambda(\xi)^4}}{\left\|G_{\lambda,L}\right\|_2^4}.$$ 
We also note $$\sum_{\xi\in A(\lambda,L)}c_\lambda(\xi)^4 \leq
\left(\sum_{\xi\in A(\lambda,L)}c_\lambda(\xi)^2\right)^2=\|G_{\lambda,L}\|_2^4.$$  At the same time Lemma \ref{lowerbound} shows that for any $\epsilon>0$
 there exists $C_\epsilon>0$ and a subsequence of density $1-\epsilon$
 such that $$1\geq\frac{\sum_{\xi\in A(\lambda,L)}c_\lambda(\xi)^4}{\|G_{\lambda,L}\|_2^4}\geq C_\epsilon$$

More precisely, if we take $\lambda$ belonging to the intersection of
the two subsequences (a subsequence of density at least $1-2\epsilon$) we
have $$1\leq \E(g_{\lambda,L}^4)\leq3-2C_\epsilon.$$  
\end{proof}


In order to conclude the proof of the theorem we need the following
approximation. 
\begin{lem}\label{approx}
Let $L = L(\lambda)
:=\tfrac{1}{20}\min(a,1/a)\lambda^{\delta/2}$.
There exists a subsequence
$\{\lambda_{j_k}\}_k$ 
  of $\Lambda = \{\lambda_j\}_j$ of density at least $1-\epsilon$
  s. t.
$$\lim_{k\to\infty}\Big|\frac{\|G_{\lambda_{j_k},L}\|_4^4}{\|G_{\lambda_{j_k},L}\|_2^4}-\frac{\|G_{\lambda_{j_k}}\|_4^4}{\|G_{\lambda_{j_k}}\|_2^4}\Big|=0.$$
\end{lem}
\begin{proof}
Recall $\E(g_{\lambda,L}^4)=\|G_{\lambda_{j_k},L}\|_4^4/\|G_{\lambda_{j_k},L}\|_2^4$.

%

There is a subsequence $\Lambda_1\subset\Lambda$ (cf. Corollary \ref{bound_trunc}) of density at least $1-\epsilon$ s.t. for $\lambda\in\Lambda_1$ we have that $\|G_{\lambda,L}\|_4/\|G_{\lambda,L}\|_2$ is
bounded from both
above and below. Moreover, there is another subsequence $\Lambda_2\subset\Lambda$ of density at least $1-\epsilon$ s.t. for $\lambda\in\Lambda_2$ we have $\|G_{\lambda_{j_k}}\|_2\gg_\epsilon 1$ (by the same argument as in \cite{KU2}, taking
$G=\epsilon^{-1}$ on p. 16).  
	
Let us denote $\Lambda_1\cap\Lambda_2=\{\lambda_{j_k}\}_{k=0}^{+\infty}$ which is a subsequence of density at least $1-2\epsilon$.

It is sufficient to show that
$$
\Big|\frac{\|G_{\lambda_{j_k},L}\|_4}{\|G_{\lambda_{j_k},L}\|_2}-
\frac{\|G_{\lambda_{j_k}}\|_4}{\|G_{\lambda_{j_k}}\|_2}\Big|
\to 0.
$$ 


We have
\begin{equation}
  \label{eq:tomato}
\begin{split}
&\Big|\frac{\|G_{\lambda_{j_k},L}\|_4}{\|G_{\lambda_{j_k},L}\|_2}-\frac{\|G_{\lambda_{j_k}}\|_4}{\|G_{\lambda_{j_k}}\|_2}\Big|\\
\leq \;&\|G_{\lambda_{j_k}}\|_2^{-1}\Big|\|G_{\lambda_{j_k}}\|_4-\|G_{\lambda_{j_k},L}\|_4\Big|\\
&+\|G_{\lambda_{j_k}}\|_2^{-1}\frac{\|G_{\lambda_{j_k},L}\|_4}{\|G_{\lambda_{j_k},L}\|_2}\Big|\|G_{\lambda_{j_k}}\|_2-\|G_{\lambda_{j_k},L}\|_2\Big|
\end{split}
\end{equation}
where we used
$\|G_{\lambda_{j_k}}\|_2\gg_\epsilon 1$, as well as Corollary
\ref{bound_trunc}, and finally the reverse triangle inequality
$|\|f\|_p-\|g\|_p|\leq\|f-g\|_p$ for $p=2,4$.

We thus find that the right hand side of (\ref{eq:tomato}) is
\begin{equation*}
\begin{split}
\ll_\epsilon \;& \|G_{\lambda_{j_k}}-G_{\lambda_{j_k},L}\|_4+\|G_{\lambda_{j_k}}-G_{\lambda_{j_k},L}\|_2 \\
\ll \;&L^{-1/12+o(1)}+L^{-1/2+o(1)}
\longrightarrow\;  0, \quad \text{as $\lambda \to \infty$,}
\end{split}
\end{equation*}
where we used Lemma \ref{L4trunc} and that
$$
\|G_{\lambda_{j_k}}-G_{\lambda_{j_k},L}\|_2^2
=
\sum_{\xi \in \cL :  ||\xi|^2-\lambda_{j_k}|\geq  L}\frac{1}{(|\xi|^2-\lambda_{j_k})^2}
\ll L^{-1+o(1)},
$$ 
which follows from an argument similar to the one used to deduce 
(\ref{eq:bound-large-v-sum}) and 
(\ref{eq:bound-small-v-sum}).
\end{proof}

If we take
$\lambda$ belonging to the subsequence of Lemma \ref{approx}, then we have for $\lambda$
sufficiently large 
$$
\left|\|g_{\lambda}\|_4^4-\|g_{\lambda,L}\|_4^4  \right| 
= o(1)
$$ 
and 
$$
\|g_{\lambda,L}\|_4^4\in[1,3-2C_\epsilon].
$$ 
Hence, it follows (recall $\E(g_\lambda^4)=\|g_\lambda\|_4^4$)
$$1+o(1)\leq \E(g_{\lambda}^4)\leq3-2C_\epsilon+ o(1).$$ 

\appendix

\section{Dirichlet boundary conditions}
In \cite{S} \u{S}eba discussed irrational aspect ratio rectangles with Dirichlet
boundary conditions rather than rectangular tori. In particular, this
means that the wave functions and the spectrum depend on the position
of the scatterer. We briefly discuss here how our results can easily
be extended to this setting.  

Let us denote the position of the scatterer by $y$. 
Denote $\cL^+=\{\xi\in\cL \mid \xi_1,\xi_2>0\}$. The new eigenfunctions are then of the form
\begin{equation}
G_\lambda(x)=\sum_{\xi\in \cL^+}c_\lambda(\xi)\psi_\xi(y)\psi_\xi(x)
\end{equation}
where $\psi_\xi(x)=\sin(2\pi \xi_1 x_1)\sin(2\pi \xi_2 x_2)$.
We note that the summation can easily be written over $\cL$:
$$
G_\lambda(x)=-\frac{1}{4}\sum_{\xi\in
  \cL}c_\lambda(\xi)\psi_\xi(y)\chi(\xi)e(\xi\cdot x),
$$ 
where $\chi(\xi)=\sgn(\xi_1)\sgn(\xi_2)$.

In order to prove the analogue of Theorem \ref{moment} we require analogues of the argument for the $L^4$-convergence in section 2.1, as well as the Lemmas \ref{L4trunc}, \ref{mom4}, \ref{lowerbound} and \ref{approx}. 

The arguments of section 2.1 and Lemma \ref{L4trunc} work analogously because of the bound $|\psi_\xi(y)|\leq 1$.

The proof of Lemma 3.3 works exactly the same way, as it only depends on the structure of the set of lattice points in the annulus $A(\lambda,L)$. In the case of Dirichlet boundary conditions it yields
$$
\E\left(\frac{G_{\lambda,L}^4}{\left\|G_{\lambda,L}\right\|_2^4}\right)
=3-2\frac{\sum_{\xi\in A(\lambda,L)}{c_\lambda(\xi)^4\psi_\xi(y)^4}}{\left\|G_{\lambda,L}\right\|_2^4}
$$
and $$\left\|G_{\lambda,L}\right\|_2^2=\sum_{\xi\in A(\lambda,L)}c_\lambda(\xi)^2\psi_\xi(y)^2$$

The analogue of Lemma \ref{lowerbound} can then be readily obtained by replacing $r_\cL(n)$
with the function $$r_\cL(n,y)=\sum_{|\xi|^2=n}\psi_\xi(y)^2\leq r_\cL(n),$$
{\em provided} we can construct a (large density) subsequence of $\Lambda$
such that  
$$\sum_{n\in I_m}\frac{r_\cL(n)}{(n-\lambda_m)^4}\ll \sum_{n\in I_m}\frac{r_\cL(n,y)}{(n-\lambda_m)^4}.$$
To do this, we define the ``bad'' set of eigenvalues
$$B=\{\lambda_k\in\Lambda' \mid \exists n\in\cN\cap I_k:
\; |\psi_\xi(y)|<\delta, |\xi|^{2}=n\}$$ 
where $\Lambda'$ denotes the subsequence of eigenvalues such that
$\#\{n \in I_m\}$ remains bounded. For
$\epsilon>0$ we may construct $\Lambda'$ of density at least
$1-\epsilon$ such that $\#\{n \in I_m\}\leq
N(\epsilon)$. 

We can now estimate the cardinality of the bad set, because for each
$n\in\cN$ such that $|\psi_\xi(y)|<\delta$ for $|\xi|^{2}=n$ there
exists only a finite 
number $K_\epsilon$ of $\lambda_k\in\Lambda'$ with $n\in I_k$. At the same time  
$$\#\{n\in\cN : n \le T,  |\psi_\xi(y)|<\delta, |\xi|^{2}=n\}=O(\delta
T)$$ so that 
$|B|=O(\delta T K_\epsilon)$ and we can make $\delta$ small enough in
terms of $\epsilon$ such that the subsequence of bad eigenvalues is of
density less than $\epsilon$. So excluding the bad eigenvalues we
obtain a subsequence of density at least $1-2\epsilon$.  

The proof of Lemma \ref{approx}, however, requires a lower bound for $\|G_{\lambda,L}\|_2$. In fact, it was already pointed out in the appendix of \cite{KU2} that for a generic position $y$, in the sense that the coordinates $y_1,y_2$ are irrational, there exists a subsequence of Laplace eigenvalues of arbitrarily high density such that for $|\xi|^2=n$ we have $\liminf_{n\to\infty}|\psi_\xi(y)|>0$. This yields the lower bound $\|G_{\lambda,L}\|_2\gg_\epsilon 1$.


\end{document}